\documentclass[reqno]{amsart}
\usepackage{mathrsfs}

\def\eps{\epsilon}%
\def\tensor{\,\raise2pt\hbox{${}_{\otimes}$}\,}
\def\fdg{\,:\,}
\def\ptl{\partial}
\def\rest#1{\raise-2pt\hbox{${\lfloor_{#1}}$}}

\def\mbo#1{\boldsymbol{#1}}

\def\olin#1{\overline{#1}{}}
\def\ulin#1{\underline{#1}{}}

\def\grad{{\nabla}}
\newcommand{\leftexp}[2]{{\vphantom{#2}}^{#1}{#2}}

\def\halb{\frac{1}{2}}

\def \a{\alpha}
\def \b {\beta}

\newtheorem{theorem}{Theorem}[section]
\newtheorem{lemma}[theorem]{Lemma}
\newtheorem{proposition}[theorem]{Proposition}
\newtheorem{corollary}[theorem]{Corollary}
\newtheorem{remark}[theorem]{Remark}
\newtheorem{definition}[theorem]{Definition}

\newcommand{\ba}{\begin{array}}
\newcommand{\ea}{\end{array}}
\newcommand{\bea}{\begin{eqnarray}}
\newcommand{\eea}{\end{eqnarray}}
\newcommand{\bee}{\begin{eqnarray*}}
\newcommand{\eee}{\end{eqnarray*}}

\renewcommand{\a}{\alpha}
\renewcommand{\b}{\beta}

\renewcommand{\r}{\rho}

\usepackage{color}

\newcommand{\green}[1]{{\color{green}#1}}

\newcounter{mnotecount}[section]

\renewcommand{\themnotecount}{\thesection.\arabic{mnotecount}}

\newcounter{mymnotecount}[section]
\renewcommand{\themymnotecount}{\thesection.\arabic{mymnotecount}}
\newcommand{\mymnote}[1]{\protect{\stepcounter{mymnotecount}}${\raisebox{0.5\baselineskip}[0pt]{\makebox[0pt][c]{\color{green}{\tiny\em$\bullet$\themnotecount}}}}$\marginpar{\raggedright\tiny\em$\!\bullet$\themymnotecount:

\green{#1}}\ignorespaces}

\renewcommand{\mymnote}[1]{}
\begin{document}

 \title[Energy for Newman-Penrose-Maxwell Scalars]{A Conserved energy for Axially Symmetric Newman-Penrose-Maxwell Scalars on Kerr Black Holes}


\author{Nishanth Gudapati}
\address{Department of Mathematics, Yale University, 10 Hillhouse Avenue, New Haven, CT-06511, USA}
\email{nishanth.gudapati@yale.edu}



\begin{abstract}
We show that there exists a 1-parameter family of positive-definite and conserved energy functionals for axially symmetric Newman-Penrose-Maxwell scalars on the maximal spacelike hypersurfaces in the exterior of Kerr black holes. It is also shown that the Poisson bracket within this 1-parameter family of energies vanishes on the maximal hypersurfaces.   
\end{abstract}

\maketitle

\section{Background and Introduction}
\noindent The Kerr metric $(\bar{M}, \bar{g})$ is a $2 -$parameter $(a, m)$ solution of the vacuum Einstein equations that represents massive, rotating black holes for $0< \vert a \vert  \leq m:$ 

\begin{align}\label{BL-Kerr}
\bar{g} =& - \left( \frac{\Delta - a^2 \sin^2 \theta}{\Sigma} \right) dt^2 - \frac{2a \sin^2 \theta (r^2 + a^2 -\Delta)}{\Sigma} dt d\phi \notag\\
&+ \left( \frac{(r^2 +a^2)^2 -\Delta a^2 \sin^2 \theta}{\Sigma}\right) \sin^2 \theta d\phi^2 + \frac{\Sigma}{\Delta} dr^2  + \Sigma d\theta^2
\end{align}
where,
\begin{subequations}
\begin{align}
\Sigma \fdg =&\, r^2 + a^2 \cos^2 \theta \\
\Delta \fdg =&\, r^2 -2mr + a^2, \quad \textnormal{with the real roots} \quad \{r_-,r_+\}
\intertext{and}
\theta \in [0, \pi],&\quad r \in (r_+, \infty),\quad \phi \in [0, 2\pi).
\end{align}

\end{subequations}
\noindent As it is evident from \eqref{BL-Kerr}, the Kerr metric admits two Killing vectors $\ptl_t$ and $\ptl_\phi.$
The problem of stability of the Kerr metric for perturbations within the class of vacuum Einstein equations is the subject of a long-standing research program in theoretical and mathematical general relativity.  Two of the important issues in the stability problem of Kerr black holes are
\begin{enumerate}
\item The lack of a \emph{positive-definite} and \emph{conserved} energy functional for the perturbations and the related superradiance effect (for $ a \neq 0$)
\item A \emph{gauge-invariant} characterization of stability. 
\end{enumerate}
In the case of  Maxwell (spin $\vert s \vert=1$) perturbations of Schwarzschild black holes (with $a=0$ in \eqref{BL-Kerr}), a positive-definite energy functional can be constructed from the energy-momentum tensor (see e.g., \cite{PB_08}). 
If one moves into the higher spin (gravitational) perturbations, even in the case of Schwarzschild black holes  $-$ which do not contain the ergo-region $-$ the construction of a positive-definite energy for the gravitational perturbations is not trivial. Using Hamiltonian methods and mode decomposition, a positive-definite energy functional for linear perturbations of the Schwarzschild black holes  was first constructed in the pioneering work of Moncrief \cite{Moncrief_74} for both even and odd parity perturbations (see also \cite{Moncrief_74_1, Moncrief_74_2, Moncrief_74_3}). In the recent complete proof of the linear stability of Schwarzschild black holes by Dafermos, Holzegel and Rodnianski \cite{HDR_16}, an important role is played by a positive-definite energy functional, which was constructed without the mode decomposition restriction (see also \cite{GH_16}). Subsequently, this energy functional was independently recovered by Prabhu-Wald \cite{PW_17_2}, by applying the methods of `canonical energy', previously constructed by Hollands-Wald \cite{WH_13}.  The linear stability based on the Cauchy problem for metric coefficients was established in \cite{HKW_16_1, HKW_16_2}. Likewise, the Morawetz estimate for  linearized gravity on Schwarzschild was established in \cite{ABW_17}, by extending the classic works \cite{Regge-Wheeler_57, Zerilli_70, Moncrief_74}.  

In the case of Kerr black holes with non vanishing angular momentum, the presence of the ergo-region causes significant difficulties in the construction of a positive-definite energy. Indeed, at the outset, it is the ergo-region and the lack of positivity of energy that results in phenomena such as the Penrose process, irreducible mass \cite{DC_70} and superradiance \cite{AAS_73, FKSY_08}. Furthermore, from a PDE perspective, the lack of a positive-definite energy poses considerable obstacles in proving asymptotic boundedness and decay of perturbations. 

A usual technique to overcome this issue is to construct a positive-definite energy functional from a linear combination of the $\ptl_t$ and $\ptl_\phi$ vector fields. However, since this energy is not necessarily conserved, a separate Morawetz or spacetime integral estimate is needed to control this energy in time. Along these lines, a variety of powerful techniques are used to prove uniform boundedness and decay of spin $s=0, 1, 2$ fields on Kerr for `small' or `very small' angular momentum \cite{LB_15_1, LB_15_2, DR_11, Tato_11, SMa_17_1, SMa_17_2, HDR_17}. Mode stability of Kerr black holes was established in the celebrated work of Whiting \cite{Whit_89}, which was recently extended in \cite{AMPW_16} to the real axis. Using spectral methods\cite{FKSY_05}, the decay of linear wave equation for fixed azimuthal modes was established in \cite{FKSY_06,
FKSY_08, FKSY_08_E} for large $\vert a \vert < m$. The decay for a general linear wave equation for large $\vert a \vert <m$ was established in \cite{DRS_16}. However, relatively little is known about the global behaviour of higher spin $(\vert s \vert=1, 2)$ fields for large $\vert a \vert.$

The special case of an axially symmetric linear wave equation admits a positive-definite energy and energy density (for $\vert  a \vert<m$) directly from the energy-momentum tensor. However, this simplification does not carry forward to Maxwell or gravitational perturbations, where counter examples for positivity of energy density can be constructed (see the discussion in Section 2 of \cite{GM17}). 
Based on the Brill mass formula for axially symmetric initial data \cite{D09}, a positive-definite energy functional for perturbations of extremal Kerr black holes was first constructed in \cite{DA_14}. Subsequently, using Hamiltonian methods, a positive-definite energy functional was constructed in \cite{GM17} for Einstein-Maxwell perturbations of Kerr-Newman black holes for the full subextremal range ($\vert a \vert, \vert Q \vert <m$).  
 A detailed discussion of the evolution of methods can be found therein. 
      

Although the Einstein's equations themselves are diffeomorphism invariant, the fact that the choice of gauge for the perturbations of the metric is not unique causes many problems in the perturbative theory.  Therefore, the characterization of perturbations of Kerr  in terms of (locally) gauge-invariant variables is crucial. 

Taking advantage of the special algebraic properties of Kerr black holes, the gauge-invariant quantities are constructed and studied in detail in several classic works. These results are summarized and streamlined in the much revered monograph of Chandrasekhar\cite{Chandrasekhar_83}. We refer the reader to this work for a detailed development of the subject. Recently, the (minimal) complete set of local gauge-invariant  perturbative quantities of Kerr black holes was obtained in \cite{AB_18}. The adjoint operators that relate the Teukolsky variables to the symmetry operators of both Maxwell and linearized gravity of Kerr are discussed in \cite{AB_17}, which builds on \cite{W_78}. In this context, it may be noted that the ergo-region, the lack of positivity of energy and superradiance also affect the dynamics of these gauge-invariant variables.

The aim of this work is to reconcile the positive-definite energy constructed in \cite{GM17} with the issue 2). In particular, we shall construct a positive-definite and conserved energy functional for the Newman-Penrose-Maxwell scalars. We would like to remark that this energy offers a significant `short cut' in the analysis of stability, in that it bypasses the need for the technical Morawetz or  spacetime integral estimates to control a positive-definite energy in time. Furthermore, the fact that the fundamental energy is of `$\Vert \cdot \Vert_{L^2}$ type' in terms of the Maxwell scalars $\Phi_0, \Phi_1, \Phi_2$  is particularly convenient in proving the explicit decay rates of the fields. The problem of establishing decay rates of perturbations using the positive-definite energy functionals is being pursued in a separate series of works. 


In this work we shall restrict to the pure Maxwell case and the case of gravitational (Einstein) perturbations of Kerr, which is a bit more technical, shall be considered in a subsequent article. Actually, the Maxwell perturbations on Kerr black holes are directly 
diffeomorphism invariant and in the case of axial symmetry, also electromagnetic-gauge invariant. Nevertheless, in view of the similarity in the structure of the Newman-Penrose scalars for Maxwell and gravitational perturbations, the motivation for the current work is that it shall serve as a prelude to the gravitational case. 

In the current article, we shall use the results of a forthcoming article \cite{GM17} for a few peripheral aspects, but the main results hold independently and are built from the foundations.  
Suppose $\ell$ and $n$ are two null vectors of $(\bar{M}, \bar{g})$ such that $\ell(n) =-1$ and let $\mbo{e_x}$ and $\mbo{e_y}$ be two (unit) orthonormal spacelike vectors, then define 
\begin{align}
m \fdg = \frac{1}{\sqrt{2}} (\mbo{e_x} + i \mbo{e_y}) \quad m^* \fdg = \frac{1}{\sqrt{2}} (\mbo{e_x} - i \mbo{e_y}).
\end{align} 
For concreteness and convenience, let us choose the Kennersly frame for the tetrad $(\ell, n, m, m^*)$ in the Newman-Penrose formalism: 
\begin{subequations}
\begin{align}
\ell \fdg =& \frac{1}{\Delta} ( (r^2 +a^2)\ptl_t + \Delta \ptl_r + a \ptl_\phi),  \\
n \fdg=& \frac{1}{2\Sigma} ((r^2+a^2)\ptl_t -\Delta \ptl_r+ a \ptl_\phi ) ,  \\
m \fdg=&\frac{1}{\olin{\Sigma} \sqrt{2}} ( ia \sin \theta \ptl_t + \ptl_\theta + i \csc  \theta \ptl_\phi ), \\
m^* \fdg =& \frac{1}{\olin{\Sigma}^* \sqrt{2}} (-ia \sin \theta \ptl_t + \ptl_\theta -i \csc  \theta \ptl_\phi ),
\end{align} 
\end{subequations}
represented in the Boyer-Lindquist coordinates $(t, r, \theta, \phi),$
where 
\begin{align}
\olin{\Sigma} \fdg = r + i a \cos \theta, \quad \olin{\Sigma}^* \fdg = r - ia \cos \theta,
\end{align}
so that, we have
\begin{align}
\bar{g}^{\mu \nu} =  -\ell^\mu n^\nu - n^\mu \ell^\nu + 
m^\mu m^{*\nu} + m^{*\mu} m^\nu \\
\ell(n) = n(\ell) =-1, \quad\text{and}\quad m^* (m)=m (m^*) =1 \label{tetrad-sign-conv},
\end{align}
and 
\[
n(n) = \ell (\ell)=0, \quad \mbo{e}_T \fdg = \frac{1}{\sqrt{2}} (\ell + n),
\]
a unit timelike vector.  In this work, we shall be interested in the Maxwell fields, governed by the Faraday tensor $F$ which is the critical point of the following functional 
\begin{align} \label{F-var}
S_{M} \fdg = -\frac{1}{4}\int \big\Vert F \big\Vert^2_{\bar{g}} \bar{\mu}_{\bar{g}}
\end{align}
for compactly supported variations of the vector potential $A$, where $F =\fdg dA.$ As a consequence, 
we also have the Bianchi identities: 

\begin{align}\label{Bianchi-orig}
\bar{\grad}_{[\gamma} F_{\mu \nu]}  =0, \quad \mu, \nu, \gamma= 0, 1 \cdots 3 \quad \textnormal{(Bianchi Identities)}
\end{align}


\noindent $\bar{\grad}$ is the covariant derivative of $(\bar{M}, \bar{g}).$ The variational principle of \eqref{F-var} results in the Maxwell field equations
\begin{align}
\bar{\grad}^\mu F_{\mu \nu} =0, \quad \textnormal{on}\quad (\bar{M}, \bar{g}), \quad \mu, \nu =0, 1, \cdots 3. 
\end{align}
The variational principle \eqref{F-var} also results in the stress-energy tensor 
\begin{align}
T_{\mu \nu} \fdg =& \frac{\ptl S_M}{\ptl g_{\mu \nu}} - \bar{g}_{\mu \nu} S_{M} \\
 =& F_{\mu \a} F^{\a}_ {\nu} - \frac{1}{4} \bar{g}_{\mu \nu} F_{\a \b} F^{\a \b}
\end{align}
which is $(\bar{M}, \bar{g})-$ divergence and trace-free, as it is well known.  
Let us now define the Newman-Penrose-Maxwell scalars in the $(n, \ell, m, \bar{m})$ tetrad as follows:
\begin{subequations}  
\begin{align}
\Phi_0 \fdg=& F_{\mu \nu} \ell^\mu m^\nu \\
\Phi_1 \fdg =& \halb F_{\mu \nu} (\ell^\mu n^\nu + m^{*\mu} m^\nu) \\
\Phi_2 \fdg =& F_{\mu \nu}  m^{*\mu} n ^\nu.
\end{align}
\end{subequations}
In general, on a globally hyperbolic, asymptotically flat manifold there are significant advantages in studying the dynamics of the Maxwell tensor $F$ using the Maxwell scalars $\Phi_0, \Phi_1, \Phi_2.$ Firstly, due to their analogous structure to the Weyl scalars, historically, the Maxwell scalars are considered to be a suitable `testing ground' to study gravitational problems. Secondly, the qualitative behaviour of the $F$ tensor is neatly separated in Maxwell scalars: $\Phi_0$, $\Phi_2$ encode
the `radiative' properties and $\Phi_1$  encodes the `Coulombic' properties of the Maxwell tensor $F$. 
Now consider a 3+1 decomposition of the Kerr metric such that $(\bar{M}, \bar{g}) = \mathbb{R} \times \olin{\Sigma},$

\begin{align}
\bar{g} = -\bar{N}^2 dt^2 + \bar{q}_{ij} (dx^i + \bar{N}^i dt) \otimes (dx^j + \bar{N}^i dt),
\end{align}
where $\bar{q}$ is the (Riemannian) metric of $\olin{\Sigma}.$
Upon a Legendre transformation of the Lagrangian action \eqref{F-var}, we get an ADM variational principle in the Hamiltonian framework,

\begin{align}
I_{ADM} \fdg = \int (A_i \ptl_t \mathfrak{E}^i - N H - N^i H_i)d ^4 x
\end{align}
for the phase space $X^{\text{Max}} \fdg = \{ (A_i, \mathfrak{E}^i), i = 1, 2, 3 \},$ where 
\begin{align}
H \fdg =& \halb \bar{\mu}^{-1}_{\bar{q}} \bar{q}_{ij} (\mathfrak{E}^i \mathfrak{E}^j + \mathfrak{B}^i \mathfrak{B}^j),
\\
H_i \fdg =& - \eps_{ijk} \mathfrak{E}^j \mathfrak{B}^k, \\
\mathfrak{B}^{i} \fdg =& \halb  \eps^{ijk} (\ptl_j A_k - \ptl_k A_j).
\end{align}

As  we already remarked, the Kerr metric $(\bar{M}, \bar{g})$ is axially symmetric with the vector $\ptl_\phi$ as the Killing field that generates the $SO(2)$ action on $(\olin{\mbo{\Sigma}}, \bar{q})$. We construct the quotient $\mbo{\Sigma}$ such that $\mbo{\Sigma} \fdg = \olin{\mbo{\Sigma}}/SO(2)$ and we denote the fixed point set of the $SO(2)$ action with $\Gamma.$ It may be noted that 
$\bar{g} (\ptl_\phi, \ptl_\phi) \equiv 0$ on $\Gamma$. Finally, define $M$ such that 
$M \fdg =\mbo\Sigma \times \mathbb{R} = \bar{M}/SO(2).$ With the above notation, define the metric $g$ on $M$ such that
\begin{align}\label{WP-Kerr}
\bar{g} = e^{-2\gamma} g + e^{2\gamma} (d\phi + \mathcal{A}_\nu dx^\nu)^2, \quad \textnormal{(Weyl-Papapetrou form)}
\end{align}
in a suitably aligned coordinate system, where $ e^{2\gamma} \fdg = \bar{g}(\ptl_\phi ,\ptl_\phi)$ and $g, \gamma, A_\nu$ are independent of $\phi.$ In explicit terms, the Kerr metric \eqref{BL-Kerr} can be represented
in the Weyl-Papapetrou form \eqref{WP-Kerr} as follows (cf. Appendix A in \cite{GM17}):  
\begin{align}
\bar{g} =& \left( \frac{\Sigma}{ (r^2+a^2)^2-a^2 \Delta \sin^2 \theta} \right) (-\Delta dt^2 + R^{-2}((r^2+a^2)^2-a^2 \Delta \sin^2 \theta) (d\r^2 + dz^2)) \notag\\
&+ \Sigma^{-1} \sin^2 \theta ((r^2+a^2)^2-a^2 \Delta \sin^2\theta ) \left(d\phi - \frac{2amr}{(r^2+a^2)^2-a^2 \Delta \sin^2 \theta}dt \right)^2,
\end{align}
where $\displaystyle R \fdg = \halb (r-m + \Delta^{1/2}), \r \fdg= R \sin \theta, z\fdg= R \cos \theta.$ 
Under the above assumptions and away from the axes $\Gamma$,  the Kerr metric satisfies the wave map equations, $L_1 =0, L_2=0$ with
\begin{align}
L_1 \fdg =& e^{2\gamma} (2(\ptl_b(N \bar{\mu}_q q^{ab} \ptl_a \gamma) + N e^{-4\gamma} \bar{\mu}_q q^{ab} \ptl_a \omega \ptl_b \omega) \\
L_2 \fdg=& - \ptl_b (N \bar{\mu}_q q^{ab} e^{-4\gamma} \ptl_a \omega)
\end{align}
upon the standard dimensional reduction procedure, where $\omega$ is the (gravitational) twist potential such that, 
\begin{align}\label{grav-twist}
\ptl_a \mathcal{A}_0 + N e^{-4\gamma} \eps_{ab} \bar{\mu}_q q^{bc} \ptl_c \omega=0.
\end{align}
 $N, q^{ab}, \bar{\mu}_q$ are such that, upon the ADM decomposition of $(M,g) = (\mbo{\Sigma}, q) \times \mathbb{R}$
\begin{align}\label{2+1adm}
g = -N^2 dt^2 + q_{ab} (dx^a + N^a dt) \otimes (dx^b + N^b dt),
\end{align}
$N$ is the lapse in \eqref{2+1adm} and $\bar{\mu}_q$ is the square root of the determinant of the metric $q_{ab}$ of $\mbo{\Sigma}.$ In this work we shall be interested in the Maxwell tensor $F$ such that it is derived from an axially symmetric $A$. In axial symmetry, we define the twist potentials $\eta, \lambda$ as follows $\lambda \fdg = A_\phi$ and from the Gauss constraint:

\begin{align} \label{eta-twist}
\mathfrak{E}^a =\fdg \eps^{ab}\ptl_b \eta, \quad (\mbo{\Sigma}, q)
\end{align}
The existence of $\eta \fdg (\mbo{\Sigma}, q) \to \mathbb{R}$ is ensured by Poincar\`e Lemma on $(\mbo{\Sigma}, q).$ We would like to emphasize that even though the Kerr manifold has non-trivial second (de Rham) cohomolgy class in the 3+1 dimensional sense, there is no need to impose a global condition for the Poincar\`e Lemma used in \eqref{eta-twist}.   This is due to the special feature of our axisymmetric problem that the quotient $(\mbo{\Sigma}, q)$ is itself a simply connected (topologically trivial) manifold, where the first cohomology class is indeed trivial.  This aspect manifests itself in several contexts in our problem. 
Equally importantly, we would like to remark that, even though we have defined the quantity $\eta$ on $\Sigma$ in the above,  it lifts up  smoothly and globally to (the Lorentzian) $(M, g)$ and transforms as a \emph{spacetime} scalar (cf. Appendix D in \cite{GM17}). Let us define $u$ and $v$ such that 
$\displaystyle u \fdg= \mathfrak{B}^{\phi}, v\fdg= -\mathfrak{E}^\phi$ so that we form the phase space 
\begin{align}
X = \fdg \{ (\lambda, v), (\eta, u) \}.
\end{align}
For convenience, let us choose $\eta=\lambda=0$ on $\Gamma$. It follows from standard arguments that global regularity holds for the initial value problem of Maxwell's equations in the domain of outer communications of Kerr black holes. As a consequence, we have $\ptl_{\mbo{n}} \lambda =\ptl_{\mbo{n}} \eta =0$ and $u=v=0$ on $\Gamma,$ where $\ptl_{\mbo{n}}$ is the derivative normal to $\Gamma.$
One approach to infer the spatial decay rate of $\eta$ from the decay rate of $\mathfrak{E}$ is shown below. It follows from the global regularity and the conditions on the axes $\Gamma$ and the horizon $\mathcal{H}^+$ that the components of $\mathfrak{E}$ admit the decomposition: 

\begin{align}
\mathfrak{E}^{\theta} = \sum^{\infty}_{n=0} \mathfrak{E}^{\theta}_{n} \cos n \theta, \quad
\mathfrak{E}^{r}= \sum^{\infty}_{n=0} \mathfrak{E}_n^{r} \sin n \theta.
\end{align}
 We have from the Gauss constraint equation, $ \ptl_r \mathfrak{E}_n^{r} - n \mathfrak{E}_n^{\theta}=0.$ The decay rate of $\eta$ can now be inferred from the equation $\ptl_a \eta = \eps_{ab} \mathfrak{E}^b$. In particular, it follows that if $\mathfrak{E}$ is compactly supported, then $(\eta, v)$ also vanish outside the support of $\mathfrak{E}$, which in turn implies the finite propagation speed of $(\eta, u).$ A similar argument applies for $(\lambda, v).$ The dynamical  field equations in $X$ can be locally represented as follows: 

\begin{subequations}\label{twist-dynamics}
\begin{align}
\ptl_t \eta =& N e^{2\gamma} \bar{\mu}^{-1}_q u,\quad 
\ptl_t u = \ptl_b (N \bar{\mu}_q q^{ab} e^{-2\gamma} \ptl_a \eta) - N \bar{\mu}_qq^{ab} e^{-4\gamma} \ptl_a \omega \ptl_b \lambda, \\
\ptl_t \lambda =& N e^{2\gamma} \bar{\mu}^{-1}_q v, \quad \ptl_t v = \ptl_b (N \bar{\mu}_q q^{ab} e^{-2\gamma} \ptl_a \lambda) + N \bar{\mu}_qq^{ab} e^{-4\gamma} \ptl_a \omega \ptl_b \eta.
\end{align}
\end{subequations}

It is well known that the Hamiltonian energy density in the phase space $X^{\text{Max}}$ has indefinite sign. In  Section 2 in \cite{GM17} it is shown that,  the Hamiltonian energy 

\begin{align}\label{Ham-orig}
H \fdg =& \int_{\mbo{\Sigma}} \Big(\halb N e^{2\gamma} \bar{\mu}^{-1}_q (u^2 + v^2) + \halb N \bar{\mu}_q q^{ab} e^{-2\gamma} (\ptl_a \eta \ptl_b \eta + \ptl_a \lambda \ptl_b \lambda) \notag\\
&\quad + N e^{-4\gamma} \bar{\mu}_q q^{ab} \ptl_a \omega \ptl_b \eta \lambda \Big) d^2x,
\end{align}
 using the transformations adapted from the Robinson's identity \cite{Rob_74},  can be transformed into a positive-definite, regularized Hamiltonian energy functional

\begin{align}\label{H-Reg-def}
H^{\text{Reg}} \fdg =& \int_{\mbo{\Sigma}}  \Big(\halb N \bar{\mu}_q^{-1} (\ulin{u}^2 + \ulin{v}^2) 
+ \halb N \bar{\mu}_q q^{ab} ( \ptl_a \gamma \ptl_b \gamma +  \frac{1}{4} e^{-4\gamma} \ptl_a \omega \ptl_b \omega) (\ulin{\lambda}^2 + \ulin{\eta}^2)  \notag\\
&+\halb N \bar{\mu}_q q^{ab} ( (\ptl_a \ulin{\lambda} - \halb \ulin{\eta} e^{-2\gamma} \ptl_a \omega) (\ptl_b \ulin{\lambda} - \halb \ulin{\eta} e^{-2\gamma} \ptl_b \omega) \notag\\
&+ (\ptl_a \ulin{\eta} + \halb \ulin{\lambda} e^{-2\gamma} \ptl_a \omega)(\ptl_b \ulin{\eta} + \halb \ulin{\lambda} e^{-2\gamma} \ptl_b \omega)) \Big) d^2 x,\notag\\
\end{align}
represented in the regularized phase space $\ulin{X} \fdg = \{ (\ulin{\lambda}, v), (\ulin{\eta}, \ulin{u}) \},$ where 
\begin{align}
\ulin{\lambda} \fdg = e^{-\gamma} \lambda, \quad \ulin{\eta} \fdg = e^{-\gamma} \eta, \quad \ulin{v} \fdg= e^{\gamma} v, \quad \ulin{u}
\fdg = e^{\gamma} u
\end{align}

\noindent such that $H^{\text{Reg}}$ is a Hamiltonian for $\ulin{X}$ i.e., 

\begin{subequations} \label{Ham-Reg}
\begin{align} 
D_{\ulin{\lambda}} \cdot H^{\text{Reg}}= - \ptl_t \ulin{v},& \quad D_{\ulin{\eta}} \cdot  H^{\text{Reg}}= - \ptl_t \ulin{u}, \\
 D_{\ulin{v}} \cdot H^{\text{Reg}} = \ptl_t \ulin{\lambda},& \quad D_{\ulin{u}} \cdot H^{\text{Reg}}= -\ptl_t \ulin{\eta}.
\end{align}
\end{subequations}


\noindent Furthermore, the aforementioned Hamiltonian $H^{\text{Reg}}$ has been used to construct a divergence-free vector field density: 
\begin{align}
(J^{\text{Reg}})^0 \fdg =&\, \mathbf{e}^{\text{Reg}} \\
(J^{\text{Reg}})^b \fdg=& - \Big(N^2 q^{ab} \ulin{u} (\ptl_a \ulin{\eta} + \halb \ulin{\lambda} e^{-2\gamma} e^{-2\gamma} \ptl_a \omega) + N^2 q^{ab} \ulin{v} (\ptl_a \ulin{\lambda}-\halb \ulin{\eta} e^{-2\gamma} \ptl_a \omega)  \Big),
\end{align}

\noindent where $\mathbf{e}^{\text{Reg}}$ is the energy density i.e., $H^{\text{Reg}} = \fdg \int_{\mbo{\Sigma}}\mathbf{e}^{\text{Reg}} \, d^2x.$
The divergence-free vector field density $J^{\text{Reg}}$ has additional information than \eqref{Ham-Reg} in that it can be used to relate the boundary fluxes through any region using the Stokes theorem. These results were later extended to the Maxwell  equations on Kerr-de Sitter in \cite{NG_17_2_ar}. In this case, the Hamiltonian contains an additional term (cf. eq (32) in \cite{NG_17_2_ar}) involving the cosmological constant $\Lambda,$ but it nevertheless generates the flow of the original Hamiltonian equations \eqref{twist-dynamics}. This is due to the special internal coupling in the equations. 

Separately, in \cite{PW_17} a 1-parameter family of energy functionals was constructed for axially symmetric Maxwell's equations on Kerr black holes. In the following, we shall reconcile their results with the Robinson's identity and also show that the energy functionals form a 1-parameter family of Hamiltonians for the dynamics in the phase space $X,$ which also shows that the Poisson bracket for different values of the parameter vanishes. 

It may be noted that the expression \eqref{H-Reg-def} is not symmetric with respect to a permutation in the phase space $X$ (or $\ulin{X}$).
If we consider an alternative form of the original Hamiltonian energy:
\begin{align}\label{Ham-orig'}
H^{\text{Alt}'}  \fdg=  \int_{\mbo{\Sigma}}& \Big( \halb N e^{2\gamma} \bar{\mu}^{-1}_q (u^2 + v^2) + \halb N \bar{\mu}_q q^{ab} e^{-2\gamma} (\ptl_a \eta \ptl_b \eta + \ptl_a \lambda \ptl_ b\lambda)\notag\\ 
&- N e^{-4\gamma} \bar{\mu}_q q^{ab} \ptl_a \omega \ptl_b \lambda \eta
\Big)d^2x
\end{align}
a modified form of the original Robinson's identity applies: 
\begin{align}\label{Rob-modified}
&\halb N e^{-2\gamma} \bar{\mu}_a q^{ab} (\ptl_a \lambda \ptl_b \lambda + \ptl_a \eta \ptl_b\eta)
- N \bar{\mu}_q q^{ab} e^{-4\gamma} \ptl_a \omega \ptl_b \lambda \eta+ \frac{1}{4}L_1 (\lambda^2 + \eta^2)+\halb L_2 \lambda \eta \notag\\
&+\halb \ptl_b \left( N \bar{\mu}_q q^{ab} e^{-4\gamma} \ptl_a \omega \eta \lambda - N \bar{\mu}_q q^{ab} e^{-2\gamma}\ptl_a \gamma (\eta^2 + \lambda^2) \right)   \notag\\
=& \frac{1}{4} N e^{-2\gamma} \bar{\mu} q^{ab} (( \ptl_a \eta +\lambda e^{-2\gamma} \ptl_a \omega) (\ptl_b \eta +\lambda e^{-2\gamma} \ptl_b \omega) + (\ptl_a \lambda - \eta e^{-2\gamma} \ptl_a\omega) ( \ptl_b \lambda - \eta e^{-2\gamma} \ptl_b \omega))  \notag\\
&+ \frac{1}{4} Ne^{-2\gamma} \bar{\mu}_a q^{ab} (( \ptl_a \lambda -2\lambda \ptl_a \gamma)( \ptl_b \lambda-2\lambda \ptl_b \gamma) + ( \ptl_a\eta -2\eta \ptl_a \gamma)( \ptl_b \eta -2\eta \ptl_b \gamma)). \notag\\
\end{align}
It may be noted that, in view of the fact that these modifications occur only in the background and divergence terms, the final energy expression remains the same as in \eqref{H-Reg-def}. However, we shall use this modification together with the original Robinson's identity to obtain a 1-parameter family of generalized Robinson's identities, which results in an energy expression that is more symmetric upon a permutation in the phase space $X$. In the process we shall recover the energy expression obtained in \cite{PW_17}.  
\begin{corollary}
Suppose $F$ is compactly supported and axially symmetric (with $\mathcal{L}_{\phi} A \equiv 0$), with smooth initial data, then the following statements hold for the initial value problem of $F$ in $(\bar{M}, \bar{g})$ with $\vert a \vert < m$:

\begin{enumerate}
\item There exists a 1-parameter family of positive-definite Hamiltonian functionals $H^{Alt}_{S} (s), s \in [0, 1]$ in the phase-space $X$, in particular, 
\begin{align}
\Big \{  H^{Alt}_{S} (s), H^{Alt}_{S} (\tau) \Big \} \equiv 0  
\end{align}

\noindent where $H^{Alt}_{S} (s)$ and $H^{Alt}_{S} (\tau)$  are such that $s \neq \tau $ with $ s, \tau \in [0, 1] $ and $\displaystyle \{ \cdot, \cdot\}$ is the Poisson bracket in the phase space $X.$

\item There exists a 1-parameter family of (spacetime) divergence-free vector field densities $J_S(s), s \in [0, 1]$ such that its flux through
$t-$constant hypersurfaces is positive-definite. 
\end{enumerate}
\end{corollary}
\begin{proof}
Consider the linear sum of the sub-Hamiltonians \eqref{Ham-orig} and \eqref{Ham-orig'} for $s \in [0, 1]$ as follows

\begin{align}
H_S(s) \fdg = &\int_{\mbo{\Sigma}} \Big( \halb N \bar{\mu}^{-1}_q e^{2\gamma} (u^2 + v^2) + \halb N \bar{\mu}_q q^{ab} e^{-2\gamma} ( \ptl_a \lambda \ptl_b \lambda + \ptl_a \eta \ptl_b \eta) \notag\\
&+ s N e^{-4\gamma} \bar{\mu}_q q^{ab} \ptl_a \omega \ptl_b \eta \lambda - (1-s) N e^{-4\gamma} \bar{\mu}_q q^{ab} \ptl_a \omega \ptl_b \lambda \eta \Big)d^2 x
\end{align}
Introduce the quantity $I(s)$ such that 
\begin{align}
I(s) \fdg =&\frac{1}{4} N \bar{\mu}_q e^{-2\gamma} q^{ab} \Big( ( \ptl_a \eta - 4 (1-s)\eta \ptl_a \gamma)( \ptl_b \eta- 4(1-s) \eta \ptl_b \gamma) \\
&+ (\ptl_a  \lambda - 4 s \lambda \ptl_a \gamma)( \ptl_b \lambda-4 s \lambda
\ptl_b \gamma)\Big) \notag\\
& + \frac{1}{4} N \bar{\mu}_q e^{-2\gamma} q^{ab} \Big(  (\ptl_a \eta + 2 s \lambda e^{-2\gamma} \ptl_a \omega)( \ptl_b \eta + 2s \lambda e^{-2\gamma} \ptl_b \omega) \notag\\
&+ (\ptl_a \lambda -2 (1-s) e^{-2\gamma} \ptl_a \omega) (\ptl_b \lambda -2 (1-s) \eta e^{-2\gamma} \ptl_b \omega) \Big) \notag\\
&- \halb N \bar{\mu}_q q^{ab} q^{ab} e^{-2\gamma} (\ptl_a \eta \ptl_b \eta + \ptl_a\lambda \ptl_b \lambda)
\intertext{and} 
II(s) \fdg =& - \ptl_b \Big( N \bar{\mu}_q q^{ab} e^{-2\gamma} \ptl_a \gamma ( s \lambda^2 + (1-s) \eta^2 ) \Big) 
\intertext{such that $I(s)-II(s)$ can be expressed as, after the imposition of the background field equations}
I(s)-II(s)=&  s N e^{-4\gamma} \bar{\mu}_q q^{ab} \ptl_a \omega \ptl_b \eta \lambda - (1-s) N e^{-4\gamma} \bar{\mu}_q q^{ab} \ptl_a \omega\ptl_b \lambda \eta \notag\\
&- 2N \bar{\mu}_q q^{ab} e^{-2\gamma} (\ptl_a \gamma \ptl_b \gamma + \frac{1}{4} e^{-4\gamma} \ptl_a \omega 
\ptl_b \omega) (s(1-2s)\lambda^2 +(1-s)(1-2(1-s))\eta^2) 
\end{align}
As a consequence, we shall transform the original Hamiltonian into the positive-definite form for $s \in [0, 1]$:
\begin{align}
H^{\text{Alt}}_S (s) \fdg =& \int_{\mbo{\Sigma}} \Big( \halb N e^{2\gamma} \bar{\mu}^{-1}_q (u^2 + v^2) +  \frac{1}{4} N \bar{\mu}_q e^{-2\gamma} q^{ab} \Big( ( \ptl_a \eta - 4 (1-s)\eta \ptl_a \gamma)( \ptl_b \eta- 4(1-s) \eta \ptl_b \gamma) \notag\\
&+ (\ptl_a  \lambda - 4 s \lambda \ptl_a \gamma)( \ptl_b \lambda-4 s \lambda
\ptl_b \gamma)\Big) + (\ptl_a \eta + 2 s \lambda e^{-2\gamma} \ptl_a \omega)( \ptl_b \eta + 2s \lambda e^{-2\gamma} \ptl_b \omega) \notag\\
&+ (\ptl_a \lambda -2 (1-s) e^{-2\gamma} \eta \ptl_a \omega) (\ptl_b \lambda -2 (1-s) \eta e^{-2\gamma} \ptl_b \omega) \Big) \notag\\
&+ 2N \bar{\mu}_q q^{ab} e^{-2\gamma} (\ptl_a \gamma \ptl_b \gamma + \frac{1}{4} e^{-4\gamma} \ptl_a \omega 
\ptl_b \omega) (s(1-2s)\lambda^2 +(1-s)(1-2(1-s))\eta^2) \Big)
d^2x \notag\\
\end{align}
where we have effectively constructed a generalized $1-$parameter family of Robinson's identities. 
We would like to remark that, interestingly, in the construction above we are not directly imposing the $L_2$ wave map equation, in contrast with \eqref{Rob-modified}
and (2.32) in \cite{GM17}. We shall now prove that $H^{\text{Alt}}_S (s)$ has the Hamiltonian structure. We recover: 
\begin{align}\label{para-kin}
D_u \cdot H^{\text{Alt}}_S(s) = N e^{2\gamma} \bar{\mu}^{-1}_q u, \quad D_{v} \cdot H^{\text{Alt}}_S(s) = N e^{2\gamma} \bar{\mu}^{-1}_q v.
\end{align}
Now consider the quantities, $D_\lambda \cdot H^{\text{Alt}}_S(s)$ and $D_\eta \cdot H^{\text{Alt}}_S(s)$ respectively. The following terms constitute $
D_\lambda \cdot H^{\text{Alt}}_S(s) :$ 
\begin{align}
N e^{-2\gamma} \bar{\mu}_q q^{ab} \ptl_a \lambda \ptl_b \lambda' =& \ptl_b ( N e^{-2\gamma} \bar{\mu}_q q^{ab} \ptl_a \lambda \lambda') - \ptl_b (N e^{-2\gamma} \bar{\mu}_q q^{ab} \ptl_a \lambda) \lambda', \notag\\
- (1-s)N \bar{\mu}_q e^{-4\gamma} q^{ab} \eta \ptl_a \lambda' \ptl_b \omega
=& - (1-s) \ptl_b (\eta N \bar{\mu}_q e^{-4\gamma} q^{ab} \ptl_b \omega \lambda') \notag\\
&+ (1-s) \ptl_b (\eta e^{-4\gamma} \bar{\mu}_q q^{ab} \ptl_a \omega) \lambda', 
\end{align}
\begin{align}
&-2s\lambda N \bar{\mu}_q e^{-2\gamma}q^{ab} \ptl_a \lambda' \ptl_b \gamma 
-2s\lambda' N \bar{\mu}_q e^{-2\gamma} q^{ab} \ptl_a \lambda \ptl_b \gamma \notag \\
&= -2s \ptl_b ( \lambda N \bar{\mu}_q e^{-2\gamma}q^{ab} \ptl_a \lambda \lambda') 
+ 2s \ptl_b (\lambda N \bar{\mu}_q e^{-2\gamma} q^{ab} \ptl_a \gamma) \lambda' \notag\\
&\quad-2s\lambda' N \bar{\mu}_q e^{-2\gamma} q^{ab} \ptl_a \lambda \ptl_b \gamma,
\end{align}
\begin{align}
sN \bar{\mu}_q e^{-4\gamma} q^{ab}\ptl_a \omega \ptl_b \eta \lambda',
\end{align}
and
\begin{align}
4N \bar{\mu}_q q^{ab} e^{-2\gamma} (\ptl_a \gamma \ptl_b \gamma + \frac{1}{4} e^{-4\gamma} \ptl_a \omega 
\ptl_b \omega) (s\lambda \lambda') 
\end{align}
where $\lambda'$ is the first variation of $\lambda.$ Likewise, $D_\eta \cdot H^{\text{Alt}}_S (s)$ is made of the terms
\begin{align}
N e^{-2\gamma} \bar{\mu}_q q^{ab} \ptl_a \eta \ptl_b \eta' =& \ptl_b (Ne^{-2\gamma} \bar{\mu}_q q^{ab} \ptl_a \eta \eta') - \ptl_b(N e^{-2\gamma} \bar{\mu}_q q^{ab} \ptl_b \eta) \eta', \notag\\
sN \lambda \bar{\mu}_q e^{-4\gamma} q^{ab} \ptl_a \eta' \ptl_b \omega =& s\ptl_b(\lambda N \bar{\mu}_q e^{-4\gamma} q^{ab} \ptl_a \omega \eta') - s \ptl_b(\lambda N \bar{\mu}_q e^{-4\gamma} q^{ab} \ptl_a \omega) \eta',
\end{align}
\begin{align}
&-2(1-s) N e^{-2\gamma} \bar{\mu}_q q^{ab}\ptl_a \eta \ptl_b \gamma \eta' 
- 2(1-s)N e^{-2\gamma} \bar{\mu}_qq^{ab} \eta \ptl_a \eta' \ptl_b \gamma \notag\\
&= -2(1-s) N e^{-2\gamma} \bar{\mu}_qq^{ab} \ptl_a \eta \ptl_b \gamma \eta' - 2(1-s) \ptl_b (\eta N e^{-2\gamma} \bar{\mu}_q 
\bar{\mu}_q q^{ab} \ptl_a \gamma \eta') \notag\\
&\quad+ 2(1-s) \ptl_b (\eta N e^{-2\gamma} \bar{\mu}_q q^{ab} \ptl_a \gamma )\eta',
\end{align}
\begin{align}
-(1-s)e^{-4\gamma} N\bar{\mu}_q q^{ab}\ptl_a \omega \ptl_b \lambda
\end{align}
and 
\begin{align}
4N \bar{\mu}_q q^{ab} e^{-2\gamma} (\ptl_a \gamma \ptl_b \gamma + \frac{1}{4} e^{-4\gamma} \ptl_a \omega 
\ptl_b \omega) (1-s)\eta \eta'
\end{align}
for the first variation $\eta'$ of $\eta.$
Collecting all the expressions above and using the background field equations, we recover the Hamiltonian field equations: 
\begin{subequations}\label{para-pot}
\begin{align}
D_\lambda \cdot H^{\text{Alt}}_S(s) =& - \ptl_b (N e^{-2\gamma} \bar{\mu}_q q^{ab} \ptl_a \lambda)  +  N \bar{\mu}_q q^{ab} e^{-4\gamma} \ptl_a \omega \ptl_b \eta, \\
D_\eta \cdot H^{\text{Alt}}_S(s) =& - \ptl_b (N e^{-2\gamma} \bar{\mu}_q q^{ab} \ptl_a \eta) -   N \bar{\mu}_q q^{ab} e^{-4\gamma} \ptl_a \omega \ptl_b \lambda.
\end{align}
\end{subequations}
In principle, if we have two conserved quantities, their Poisson bracket provides another conserved quantity. 
However, it follows immediately from \eqref{para-kin} and \eqref{para-pot}, that the Poisson bracket 
\begin{align}
\Big \{  H^{\text{Alt}}_{S} (s), H^{\text{Alt}}_{S} (\tau) \Big \} \equiv 0  
\end{align}
for any fixed $s, \tau \in [0, 1], s \neq \tau.$ In other words, the $1-$parameter family $H^{\text{Alt}}_{S} (s), s \in [0, 1]$ are in involution. 

 If we consider the phase space $\ulin{X}$ we can transform the aforementioned Hamiltonian energy density as follows
 \begin{align}
 &\frac{1}{4}N e^{-2\gamma} \bar{\mu}_q q^{ab} \big((\ptl_a \lambda -4s\lambda \ptl_a\gamma) (\ptl_b \lambda-4s\lambda \ptl_b \gamma) + (\ptl_a \eta -4(1-s)\eta \ptl_a\gamma) (\ptl_b \eta-4(1-s)\eta \ptl_b \gamma) \big) \notag\\
  &+ \frac{1}{4} N e^{-2\gamma} \bar{\mu}_q q^{ab} \big(( \ptl_a \eta + 2s \lambda e^{-2\gamma} \ptl_a \omega)(\ptl_b \eta + 2s\lambda e^{-2\gamma} \ptl_b \omega) \notag \\
  &+ (\ptl_a \lambda - 2(1-s)\eta e^{-2\gamma} \ptl_a \omega)(\ptl_b \lambda - 2(1-s) \eta e^{-2\gamma} \ptl_b \omega) \big)  \notag\\
  &+ 2N\bar{\mu}_q q^{ab} e^{-2\gamma} (\ptl_a \gamma \ptl_b \gamma + \frac{1}{4}e^{-4\gamma} \ptl_a \omega \ptl_b \omega) (s(1-2s) \lambda^2 + (1-s) (1-2(1-s)) \eta^2) \notag\\
 =\quad&\halb N \bar{\mu}_q e^{-2\gamma} q^{ab} \Big( (\ptl_a \eta - 2 (1-s) \eta \ptl_a \gamma + s \lambda e^{-2\gamma} \ptl_a \omega) (\ptl_b \eta - 2 (1-s) \eta \ptl_b \gamma + s \lambda e^{-2\gamma} \ptl_b \omega) \notag\\
 &+ (\ptl_a \lambda -2s \lambda \ptl_a \gamma - (1-s) \eta e^{-2\gamma} \ptl_a \omega) (\ptl_b \lambda -2s \lambda \ptl_b \gamma - (1-s) \eta e^{-2\gamma} \ptl_b \omega)  \Big) \notag\\
 &+ 2Ns(1-s) \bar{\mu}_q q^{ab} e^{-2\gamma} (\ptl_a \gamma \ptl_b \gamma + \frac{1}{4} e^{-4\gamma} \ptl_a \omega 
\ptl_b \omega) ( \lambda^2 + \eta^2) \\
 =\quad& \halb N \bar{\mu}_q q^{ab} \Big( (\ptl_a \ulin{\lambda} -2 (s-\halb) \ulin{\lambda}\ptl_a \gamma- (1-s) \ulin{\eta} e^{-2\gamma} \ptl_a \omega) (\ptl_b \ulin{\lambda} -2 (s-\halb) \ulin{\lambda} \ptl_b \gamma- (1-s) \ulin{\eta} e^{-2\gamma} \ptl_b \omega) \notag\\
  &+ (\ptl_a \ulin{\eta} - 2(\halb-s)\ulin{\eta}\ptl_a \gamma+ s \ulin{\lambda} e^{-2\gamma} \ptl_a \omega)(\ptl_b \ulin{\eta} - 2(\halb-s)\ulin{\eta}\ptl_a \gamma+ s \ulin{\lambda} e^{-2\gamma} \ptl_b \omega) \Big) \notag\\
  &+ 2 s(1-s)N \bar{\mu}_q q^{ab}( \ptl_a \gamma \ptl_b \gamma +  \frac{1}{4} e^{-4\gamma} \ptl_a \omega \ptl_b \omega) (\ulin{\lambda}^2 + \ulin{\eta}^2).
\end{align}

\noindent So that we have the expression, 

\begin{align}
H^{\text{Reg}}_S(s) \fdg =& \int_{\mbo{\Sigma}} \Big( \halb N \bar{\mu}^{-1}_q (\ulin{u}^2 + \ulin{v})^2+ 2 s(1-s)N \bar{\mu}_q q^{ab}( \ptl_a \gamma \ptl_b \gamma +  \frac{1}{4} e^{-4\gamma} \ptl_a \omega \ptl_b \omega) (\ulin{\lambda}^2 + \ulin{\eta}^2) \notag\\
& \halb N \bar{\mu}_q q^{ab} \Big( (\ptl_a \ulin{\lambda} -2 (s-\halb) \ulin{\lambda}\ptl_a \gamma- (1-s) \ulin{\eta} e^{-2\gamma} \ptl_a \omega) (\ptl_b \ulin{\lambda} -2 (s-\halb) \ulin{\lambda} \ptl_b \gamma- (1-s) \ulin{\eta} e^{-2\gamma} \ptl_b \omega) \notag\\
  &+ (\ptl_a \ulin{\eta} - 2(\halb-s)\ulin{\eta}\ptl_a \gamma+ s \ulin{\lambda} e^{-2\gamma} \ptl_a \omega)(\ptl_b \ulin{\eta} - 2(\halb-s)\ulin{\eta}\ptl_a \gamma+ s \ulin{\lambda} e^{-2\gamma} \ptl_b \omega) \Big) \Big)d^2x 
\end{align}

\noindent which also serves as a Hamiltonian for the dynamics of $\ulin{X}$ i.e., 
\begin{align}\label{Ham-S}
D_{\ulin{\lambda}} \cdot H^{\text{Reg}}_S(s)= - \ptl_t v,& \quad D_{\ulin{\eta}} \cdot  H^{\text{Reg}}_S(s)= - \ptl_v u, \\
 D_{\ulin{v}} \cdot H^{\text{Reg}}_S(s)= \ptl_t \lambda,& \quad D_{\ulin{u}} \cdot H^{\text{Reg}}_S(s)= -\ptl_t \eta.
\end{align}
Upon appropriate adjustment of notation, this energy functional matches with the one obtained in \cite{PW_17}. 
Let us now calculate the $\frac{\ptl}{\ptl t} \mathbf{e}_S^{\text{Reg}} (s),$ where $\mathbf{e}_S^{\text{Reg}} (s)$ is the energy density i.e., $H_S^{\text{Reg}} (s) = \int_{\mbo{\Sigma}} \mathbf{e}_S^{\text{Reg}} (s) d^2x.$ Define the quantities $\bar{u} \fdg = N \bar{\mu}^{-1}_q u$
and $\bar{v} \fdg = N \bar{\mu}^{-1}_q v$, then the $\ptl_a (\ptl_t \eta)$  and $\ptl_a (\ptl_t \lambda)$ terms can be represented as
\begin{align}
N \bar{\mu}_q q^{ab} \ptl_a \bar{u} (\ptl_a \eta -2(1-s) \eta \ptl_b \gamma + s \lambda e^{-2\gamma} \ptl_b \gamma) \notag\\+ 2 \bar{u}
N \bar{\mu}_q q^{ab} \ptl_a \gamma (\ptl_a \eta -2(1-s) \eta \ptl_b \gamma + s \lambda e^{-2\gamma} \ptl_b \gamma)
\intertext{and}
N \bar{\mu}_q q^{ab} \ptl_a \bar{v} (\ptl_b \lambda -2s \lambda \ptl_b \gamma -(1-s) e^{-2\gamma} \ptl_b \omega) \notag\\
2 \bar{v}N \bar{\mu}_q q^{ab} \ptl_a \gamma (\ptl_b \lambda -2s \lambda \ptl_b \gamma -(1-s) e^{-2\gamma} \ptl_b \omega)
\end{align}
respectively. Likewise, the $\ptl_t \eta$ and $\ptl_t \lambda$ terms can be represented as 
\begin{align}
&\bar{u}\Big( \ptl_b (N \bar{\mu}_q q^{ab} \ptl_a \eta) + N \bar{\mu}_q q^{ab} (-2 \ptl_a \eta \ptl_b \gamma + e^{-2\gamma}\ptl_a\omega \ptl_b \lambda -2 (1-s) \ptl_a \gamma (\ptl_b \eta -2(1-s) \eta \ptl_b \gamma + s \lambda e^{-2\gamma} \ptl_b \omega) \notag\\
&- (1-s) e^{-2\gamma} \ptl_a \omega ( \ptl_b \lambda -2s \lambda \ptl_b \gamma -(1-s) \eta e^{}-2\gamma) \ptl_b \omega +4s(1-s) (\ptl_a \gamma \ptl_b \gamma + \frac{1}{4} e^{-4\gamma} \ptl_a \omega \ptl_b \omega)\eta ) \Big)
\intertext{and}
& \bar{v} \Big( \ptl_b (N \bar{\mu}_q q^{ab} \ptl_b \lambda) + N\bar{\mu}_q q^{ab} (-2 \ptl_a \gamma \ptl_b \lambda+ e^{-2\gamma}\ptl_a \omega \ptl_b \eta  + s \lambda e^{-2\gamma} \ptl_a \omega (\ptl_a \eta
-2(1-s) \eta \ptl_b \gamma + s \lambda e^{-2\gamma} \ptl_b \omega)\notag\\
&-2s \ptl_a \gamma (\ptl_a \lambda -2s \lambda \ptl_b \gamma - (1-s) \eta e^{-2\gamma} \ptl_b \omega) + 4s(1-s) (\ptl_a \gamma \ptl_b \gamma + e^{-4\gamma} \ptl_a \omega \ptl_b \omega)\lambda) \Big).
\end{align}

\noindent Collecting the $\ptl_a \bar{u}, \bar{u}$ and $\ptl_a \bar{v}, \bar{v}$ separately in the above, we get

\begin{align}
\frac{\ptl}{\ptl t}\mathbf{e}_S^{\text{Reg}} (s) =&\, \ptl_b \Big(N \bar{\mu}_q q^{ab}\bar{u} (\ptl_a \eta -2(1-s) \eta \ptl_b \gamma + s \lambda e^{-2\gamma} \ptl_b \omega)  \notag\\
&+ N \bar{\mu}_q q^{ab} \bar{v} ( \ptl_a \lambda
-2s\lambda \ptl_a \gamma - (1-s) \eta e^{-2\gamma} \ptl_a \omega)\Big) \notag\\
=&\, \ptl_b \Big(N^2 q^{ab} u (\ptl_a \eta -2(1-s) \eta \ptl_b \gamma + s \lambda e^{-2\gamma} \ptl_b \omega)  \notag\\
&+ N^2 q^{ab} v ( \ptl_a \lambda
-2s\lambda \ptl_a \gamma - (1-s) \eta e^{-2\gamma} \ptl_a \omega)\Big).
\end{align}
Therefore, the vector field density $J^{\text{Reg}}_S(s)$ defined as 
\begin{subequations}
\begin{align}
(J^{\text{Reg}}_S(s))^t \fdg=& \mathbf{e}^{\text{Reg}}_S(s) \\
(J^{\text{Reg}}_S(s))^a \fdg=&-N^2 q^{ab} \Big(u (\ptl_a \eta -2(1-s) \eta \ptl_b \gamma + s \lambda e^{-2\gamma} \ptl_b \omega)  \notag\\
&+ v ( \ptl_a \lambda
-2s\lambda \ptl_a \gamma - (1-s) \eta e^{-2\gamma} \ptl_a \omega) \Big)
\end{align}
\end{subequations}
is (spacetime) divergence free. As we already noted, the divergence-free $J^{\text{Reg}}_S$ has additional information than \eqref{Ham-S} in that it can be used to relate the fluxes through various hypersurfaces, without a bulk term. 
\end{proof}


\section{A Conserved Energy For Newman-Penrose-Maxwell Scalars}

 For convenience, let us now represent the tetrad 1-forms in Boyer-Lindquist coordinates, that are consistant with the normalization introduced above: 
\begin{subequations}
\begin{align}
\ell =& \frac{1}{\Delta} (-\Delta dt +\Sigma dr +a \sin^2\theta \Delta d\phi ), \\
n =& \frac{1}{2\Sigma} (-\Delta dt - \Sigma dr +a \sin^2 \theta \Delta d\phi), \\
m =& \frac{1}{\olin{\Sigma} \sqrt{2}} (-i a \sin\theta dt + \Sigma d\theta + i (r^2+a^2) \sin \theta d\phi), \\
m^*=& \frac{1}{\olin{\Sigma}^* \sqrt{2}} (i a \sin\theta dt + \Sigma d\theta - i (r^2+a^2) \sin \theta d\phi),
\end{align} 
\end{subequations}
so that,
\begin{align}\label{tetrad-low}
\bar{g}_{\mu \nu} = - \ell_\mu n_\nu - n_\mu \ell_\nu + m_{\mu}m^*_\nu + m^*_\mu m_\nu.
\end{align}
In this work, we shall use the following convention for the anti-symmetric sum $X_{[a}Y_{b]} \fdg = X_a Y_b - X_b Y_a$ (i.e., without the factor of 2). 
Upon the inversion of basis and taking advantage of the tetrad form \eqref{tetrad-low}, the  $\mathfrak{E}$ and 
$\mathfrak{B}$ fields can be represented in terms of the Maxwell scalars as follows: 

\begin{align}
\mathfrak{E}^i =& 2e^{-2\gamma}N \bar{\mu}_{q} \left( \text{Re} ( \Phi_0 m^{*[0} n^{i]}+ \Phi_1 (n^{[0} \ell^{i]} + m^{[0} m^{*i]}) + \Phi_2 \ell^{[0} m^{i]}  )\right),\\
\mathfrak{B}^i =& \eps^{ijk} \text{Re} (\Phi_0 m^{*}_{[j} n_{k]}+ \Phi_1 (n_{[j} \ell_{k]} + m_{[j} m_{k]}^{*}) + \Phi_2 \ell_{[j} m_{k]}  ),\quad i,j = 1,2, 3.
\end{align}
where $\text{Re}(z) = 2^{-1} (z + z^*).$
For later use, let us collect the following quantities in Boyer-Lindquist coordinates: 

\begin{subequations}
\begin{align}
\ell^{[0} m^{3]} =& \frac{i}{\sqrt{2} \olin{\Sigma} \Delta} ( -a^2 \sin^2 \theta +  \csc \theta (r^2 + a^2))\\
 m^{*[0} n^{3]}=& \frac{i}{2 \sqrt{2} \Sigma \olin{\Sigma}} (\csc \theta (r^2 +a^2) - a^2 \sin \theta ) \\
\ell^{[1} m^{2]} =& \frac{1}{\sqrt{2} \olin{\Sigma}} , \quad m^{*[1} n^{2]}= -\frac{\Delta}{2 \sqrt{2} \Sigma \olin{\Sigma}^*}  \\
\ell^{[0} n^{3]} =& 0, \quad m^{*[0} m^{3]}= 0, \quad\ell^{[1} n^{2]} = 0, \quad m^{*[1} m^{2]}= 0.
\end{align}
\end{subequations}

In the following, we shall represent the Maxwell scalars $\Phi_0, \Phi_1, \Phi_2$ in terms of the phase space variables $X = \{ (\lambda, v), (\eta, u) \}$ and 
dimensionally reduced form. 

\begin{subequations}\label{Maxwell-scalars-twist}
\begin{align}
\Phi_0 =& N e^{2\gamma} \bar{\mu}_q v \frac{i}{\sqrt{2} \olin{\Sigma} \Delta} ( -a^2 \sin^2 \theta +  \csc \theta (r^2 + a^2)) + u  \frac{1}{\sqrt{2} \olin{\Sigma}} \notag\\
&+ \left(\frac{(-N^2 e^{-2\gamma} + e^{2\gamma} \mathcal{A}^2_0)N^{-1} \bar{\mu}^{-1}_q q_{ab} \eps^{bc} \ptl_c \eta}{1-N^{-2} e^{4\gamma} \mathcal{A}^2_0}-\mathcal{A}_0 \ptl_a \lambda \right) \ell^{[0} m^{a]} + \ptl_a \lambda \ell^{[a} m^{3]}, \\ 
\Phi_1=& \halb \left(\frac{(-N^2 e^{-2\gamma} + e^{2\gamma} \mathcal{A}^2_0)N^{-1} \bar{\mu}^{-1}_q q_{ab} \eps^{bc} \ptl_c \eta}{1-N^{-2} e^{4\gamma} \mathcal{A}^2_0}-\mathcal{A}_0 \ptl_a \lambda \right) (\ell^{[0} n^{a]} + m^{*[0}m^{a]})\notag\\
 &+ \halb\ptl_a \lambda (\ell^{[a} m^{3]} + m^{*[a}m^{3]}), \\
 \intertext{and}
\Phi_2 =& N e^{2\gamma} \bar{\mu}_q v \frac{i}{2\sqrt{2} \Sigma \olin{\Sigma}} ((r^2+a^2)\csc \theta -a^2 \sin \theta) +  u \frac{-\Delta}{2\sqrt{2}\Sigma \olin{\Sigma}^*} \notag\\
&+ \left(\frac{(-N^2 e^{-2\gamma} + e^{2\gamma} \mathcal{A}^2_0)N^{-1} \bar{\mu}^{-1}_q q_{ab} \eps^{bc} \ptl_c \eta}{1-N^{-2} e^{4\gamma} \mathcal{A}^2_0}-\mathcal{A}_0 \ptl_a \lambda \right) m^{*[0} n^{a]} + \ptl_a \lambda m^{*[a} n^{3]}.
\end{align}
\end{subequations}
\subsection*{Derivative Operators and Spin Coefficients}
Let us define the (directional) derivative operators along the tetrad $(\ell, n, m, m^*)$ as follows 
\begin{align}
\mbo{D} \fdg = \ell^\mu \ptl_\mu, \quad \mbo{\Delta} \fdg = n^\mu \ptl_\mu, \quad \mbo{\delta} \fdg = m^\mu \ptl_\mu, \quad \mbo{\delta}^* \fdg = \bar{m}^\mu \ptl_\mu.
\end{align}
In consistancy with our Hamiltonian framework, we had to chose the $(-+++)$ sign convention for our metric. As a consequence, the null tetrad has `$(--++)$' sign convention (cf. \eqref{tetrad-sign-conv}),
which in turn alters the definitions of 
spin coefficients from the standard literature (e.g., \cite{Chandrasekhar_83}). 
We shall now define spin coefficients from first principles and evaluate them for the Kerr metrics as per our conventions, for the convenience of the reader. We shall also derive the Maxwell's equations for $\Phi_0, \Phi_1, \Phi_2$ accordingly.  
\begin{subequations}\label{spin-coeff}
\begin{align}
\mbo{\r} \fdg =&- m^\mu m^{*\nu} \bar{\grad}_\nu \ell_\mu =  \frac{1}{\olin{\Sigma}^*},  \\
 \mbo{\tau} \fdg =&- m^\mu n^\nu \bar{\grad}_\nu \ell_\mu
= \frac{i a \sin \theta}{\sqrt{2}\Sigma}, \\
\mbo{\mu} \fdg=&  m^{*\mu} m^\nu \bar{\grad}_\nu n_{\mu} =  \frac{\Delta}{2 \olin{\Sigma}^* \Sigma}, \\
\mbo{\pi} \fdg=&  m^{* \mu} \ell^\nu \bar{\grad}_\nu n_{\mu} =- \frac{i a \sin \theta}{\olin{\sqrt{2}\Sigma}^{* 2}}, \\
\mbo{\gamma} \fdg=& \halb ( -n^\mu n^\nu \bar{\grad}_\nu \ell_\mu + m^{*\mu} n^\nu \bar{\grad}_\nu m_\mu ) = \frac{\Delta}{2 \olin{\Sigma}^* \Sigma} - \frac{r-m}{2 \Sigma} ,\\
\mbo{\b} \fdg = & \halb (-n^\mu m^\nu \bar{\grad} _\nu \ell_\mu + m^{*\mu} m^\nu \bar{\grad}_\nu m_{\mu}) =- \frac{\cot \theta}{2 \sqrt{2}\,\olin{\Sigma}}, \\
\mbo{\a} \fdg=& \halb (-n^\mu m^\nu \bar{\grad}_\nu \ell_\mu + m^{*\mu} m^{*\nu} \bar{\grad}_\nu m_\mu) =- \frac{i a \sin \theta}{\olin{\sqrt{2}\Sigma}^{* 2}} + \frac{\cot \theta}{2 \sqrt{2} \olin{\Sigma}^*}.
\end{align}
\end{subequations}
From the definitions and in view of the fact that the Kerr metric is of Petrov type D, we have 
\begin{align}
&\mbo{\kappa} \fdg=- \ell^\mu m^\nu \bar{\grad}_\nu \ell_\mu \equiv \mbo{\sigma} \fdg = -m^\mu m^\nu \bar{\grad}_\nu \ell_\mu \equiv \mbo{\lambda} \fdg= m^{*\mu} m^{*\nu} \bar{\grad}_\nu n_{\mu} \equiv 0,\notag\\
& \mbo{\nu} \fdg = 
m^{*\mu} n^{\nu} \bar{\grad}_\nu n_{\mu}  \equiv \mbo{\eps} \fdg = \halb (- n^\mu \ell^\nu \bar{\grad}_\nu \ell_\mu + m^{*\mu} \ell^\nu \bar{\grad}_\nu m_{\mu})   \equiv 0
\end{align}
 for Kerr metrics. 
\subsection*{Maxwell's equations}
The Maxwell field equations 
\begin{align}
\bar{\grad}^\mu F_{\mu \nu} =0
\end{align}
can be written in the tetrad form as, 

\begin{subequations}
\begin{align}
-\bar{\grad}_{\ell} F_{\ell n} + \bar{\grad}_{m^*} F_{\ell m} + \bar{\grad}_{m} F_{\ell m^*} =&0 \\
-\bar{\grad}_n F_{n \ell} + \bar{\grad}_{m^*} F_{n m} + \bar{\grad}_{m} F_{n m^*} =&0 \\
-\bar{\grad}_{n} F_{m \ell} - \bar{\grad}_{\ell} F_{m n} + \bar{\grad}_m F_{m m^*}=&0 \\
-\bar{\grad}_n F_{m^* \ell} - \bar{\grad}_\ell F_{m^*n}+ \grad_{m^*} F_{m^*m} =& 0
\end{align}
\end{subequations}

\noindent Using the Bianchi identities \eqref{Bianchi-orig}, we get correspondingly
\begin{subequations}\label{Maxwell-cova}
\begin{align}
\bar{\grad}_\ell \Phi_1= \bar{\grad}_{m^*} \Phi_0,& \quad \bar{\grad}_m \Phi_2 = \bar{\grad}_{n} \Phi_1,\\
\bar{\grad}_\ell \Phi_2= \bar{\grad}_{m^*} \Phi_1,& \quad\bar{\grad}_{m}\Phi_1= \bar{\grad}_{n} \Phi_0
\end{align}
\end{subequations}
in the clockwise order. 
\noindent Now, eliminating the covariant derivatives acting on Maxwell scalars in favour of the
directional derivatives, we get the following
\begin{align}
\bar{\grad}_\a \Phi_1= \ptl_\a \Phi_1 - m^\mu (e_\a)^\nu \bar{\grad}_\nu \ell_\mu \Phi_2 - m^{*\mu} (e_\a)^\nu \bar{\grad}_\nu n_{\mu} \Phi_0,
\end{align}

\begin{subequations}
\begin{align}
\bar{\grad}_{m^*} \Phi_0= \mbo{\delta}^* \Phi_0 + 2\mbo{\r} - 2\mbo{\alpha} \Phi_1,& \quad \bar{\grad}_\ell \Phi_1 = \mbo{D} \Phi_1 + \mbo{\kappa} \Phi_2 - \mbo{\pi} \Phi_0, \\
\bar{\grad}_m \Phi_2 = \mbo{\delta} \Phi_2+ 2 \mbo{\b} \Phi_2 -2 \mu \Phi_1,& \quad \bar{\grad}_n \Phi_1 = \mbo{\Delta} \Phi_1 + \mbo{\tau} \Phi_2 - \mbo{\nu} \Phi_0, \\
\bar{\grad}_n \Phi_0= \Delta \Phi_0 + 2 \mbo{\tau} \Phi_1 -2 \mbo{\gamma} \Phi_0,& \quad\bar{\grad}_m \Phi_1 = \mbo{\delta} \Phi_1 + \mbo{\sigma} \Phi_2 - \mbo{\mu} \Phi_0, \\
\bar{\grad}_\ell \Phi_2 = \mbo{D} \Phi_2 + 2 \mbo{\eps} \Phi_2-2 \mbo{\pi} \Phi_1,& \quad \bar{\grad}_{m^*} \Phi_1 = \mbo{\delta}^* \Phi_1 + \mbo{\r} \Phi_2 - \mbo{\lambda} \Phi_0.
\end{align}
\end{subequations}

\noindent Consequently, the Maxwell's field equations on Kerr metrics are

\begin{subequations}\label{Maxwell-eqn}
\begin{align}
\mbo{D} \Phi_1 - \mbo{\delta}^* \Phi_0 =& (\mbo{\pi} -2\mbo{\a}) \Phi_0 + 2\mbo{\r} \Phi_1,  \\
\mbo{\Delta} \Phi_1 - \mbo{\delta} \Phi_2 =& (2 \mbo{\b} - \mbo{\tau}) \Phi_2-2\mbo{\mu} \Phi_1,  \\
  \mbo{\delta} \Phi_1-\mbo{\Delta} \Phi_0 =& (\mbo{\mu-}2\mbo{\gamma} ) \Phi_0 +2\mbo{\tau} \Phi_1,\\
\mbo{\delta}^* \Phi_1-\mbo{D} \Phi_2  =& (2\mbo{\eps}-\mbo{\r}) \Phi_2-  2\mbo{\pi} \Phi_1  
 \end{align}
\end{subequations}
respectively, where the spin coefficients for the Kerr metric are defined and expressed as in \eqref{spin-coeff}. In the case of axial symmetry, if we use the formulas \eqref{Maxwell-scalars-twist}, the satisfaction of system \eqref{Maxwell-eqn}, in consistancy with the field equations \eqref{twist-dynamics}, is readily verified. 

\begin{proposition}
Suppose $F= dA$ is a Maxwell tensor that satisfies the Maxwell's equations and is axially symmetric $\mathcal{L}_\phi F \equiv 0.$ Then, 
\begin{enumerate}
\item The Maxwell scalars $\Phi_0, \Phi_1, \Phi_2$ are also axially symmetric $\mathcal{L}_\phi \Phi_i \equiv 0, i=0, 1, 2$
\item Suppose, $\psi = \Phi_0$  or $ \olin{\Sigma}^{*2} \Phi_2$, then  
\begin{align} \label{Teukolsky-axi}
&\left(\frac{(r^2+a^2)^2}{\Delta} -a^2 \sin^2 \theta \right) \ptl^2_t \psi 
- \Delta^{-s} \ptl_r (\Delta^{s+1} \ptl_r \psi) - \frac{1}{\sin \theta} \ptl_\theta (\sin \theta \ptl_\theta \psi) \notag\\
 &-2s \left(\frac{m (r^2-a^2)}{\Delta} -r -ia \cos\theta \right) \ptl_t \psi + (s^2 \cot^2 \theta -s) \psi =0
\end{align}
\end{enumerate}
for $s= \pm 1.$
\end{proposition}
\begin{proof}
Part 1) follows  by inspection, while noting that our tetrad is also axially symmetric and part 2) is the famous Teukolsky's master equation \cite{Teukolsky_73} with axial symmetry, for which, the case $\vert s \vert=1$ corresponds to Maxwell's equations. The `extreme' components $\Phi_0$ and $\Phi_2$ are also related  by the celebrated Teukolsky-Starobinsky differential identities. 
\end{proof}
We would like to remark that the Maxwell perturbations are governed by the two independent degrees of freedom, corresponding to the Maxwell scalars $\Phi_0$ and $\Phi_2.$ However, the transformation of the Maxwell field equations to the field equations (Teukolsky's equation \eqref{Teukolsky-axi}) for these extreme components are governed by the higher order differential operators. In this work we shall focus on the \emph{total}  energy of the \emph{fundamental} Maxwell field equations \eqref{Maxwell-eqn}, so it involves all the Maxwell scalars. 
It may be noted that the lack of positivity of energy also affects the dynamics of the Maxwell scalars $\Phi_0, \Phi_1, \Phi_2.$ This is evident if we represent the original Maxwell energy, corresponding to the Hamiltonian flow of $\ptl_t,$ in terms of the axially symmetric NP scalars on the Kerr metric:

 \begin{align}\label{Bel-Rob-Max}
E(\Phi_0,& \Phi_1, \Phi_1) \notag\\
\fdg= \int_{\mbo{\Sigma}} \Big( &\halb N e^{2\gamma} \bar{\mu}^{-1}_q \Big(  (2e^{-2\gamma}N \bar{\mu}_{q} ( \text{Re} ( \Phi_0 m^{*[0} n^{3]}+ \Phi_1 (n^{[0} \ell^{3]} + m^{[0} m^{*3]})\notag\\ 
+&\Phi_2 \ell^{[0} m^{3]}  )))^2 
 +  (2\text{Re} (\Phi_0 m^{*}_{[1} n_{2]}+ \Phi_1 (n_{[1} \ell_{2]} + m_{[1} m_{2]}^{*}) + \Phi_2 \ell_{[1} m_{2]}  ))^2 \Big) \notag\\
&\halb N \bar{\mu}_q q^{ab} e^{-2\gamma} \Big( (2 \text{Re}(\Phi_0 m^*_{[a}n_{3]}+  \Phi_1 (n_{[a} \ell_{3]} + m_{[a} m^*_{3]}) + \Phi_2 \ell_{[a} m_{3]} ) )\notag\\
&(2\text{Re}(\Phi_0 m^*_{[b}n_{3]}+ \Phi_1 (n_{[b} \ell_{3]} + m_{[b} m^*_{3]}) + \Phi_2 \ell_{[b} m_{3]} ) \notag\\
 +& ( 2 e^{-2\gamma}N \bar{\mu}_q \eps_{ac} \text{Re}( \Phi_0 m^{[*0} n^{c]} + \Phi_1 (n^{[0}\ell^{c]}+ m^{[0} m^{*c]}) + \Phi_2 \ell^{[0}m ^{c]} ) \notag\\
 & ( 2 e^{-2\gamma}N \bar{\mu}_q \eps_{bc} \text{Re}( \Phi_0 m^{[*0} n^{c]} + \Phi_1 (n^{[0}\ell^{c]}+ m^{[0} m^{*c]}) + \Phi_2 \ell^{[0}m ^{c]} ) \Big) \notag\\
  -&\bar{N}^\phi \eps^{ab} \Big( ( 2\text{Re}(\Phi_0 m^*_{[a}n_{3]}+  \Phi_1 (n_{[a} \ell_{3]} + m_{[a} m^*_{3]}) + \Phi_2 \ell_{[a} m_{3]} ) \notag\\
 &  ( 2 e^{-2\gamma}N \bar{\mu}_q \eps_{bc} \text{Re}( \Phi_0 m^{[*0} n^{c]} + \Phi_1 (n^{[0}\ell^{c]}+ m^{[0} m^{*c]}) + \Phi_2 \ell^{[0}m ^{c]} ) \Big) \Big) d^2x.
\end{align}

In the following we shall construct a positive-definite
and conserved energy functional using a non-local canonical transformation from the twist potential variables. 

We would like to remark that the energy expression \eqref{Bel-Rob-Max} has a similar structure to the original Bel-Robinson energy of the Weyl scalars corresponding to the gravitational perturbations (cf. Appendix I in \cite{GM17}). However, in contrast with the Maxwell case, Weyl scalars differ in two orders of derivatives from the twist potential variables used in the construction of the positive-definite energy functional for gravitational perturbations (which in turn is closely related to the ADM mass). 

\begin{theorem}

\begin{enumerate}
\item 
Suppose, $\eta \fdg (\mbo{\Sigma}, q) \to \mathbb{R}$ and $\lambda \fdg (\mbo{\Sigma}, q) \to \mathbb{R}$ are the twist potentials such that $\mathfrak{E}^a = \eps^{ab}\ptl_b \eta$ and $\mathfrak{B}^a = \eps^{ab} \ptl_b \lambda$
then $\eta$ and $\lambda$ are uniquely given by 
\begin{align}
\eta =&\,   \ptl_b \big(2Ne^{2\gamma} \bar{\mu}^2_q q^{ab} \eps_{ca} \textnormal{Re} (\Phi_0 m^{*[0} n^{c]} + \Phi_1 (n^{0} \ell^{c]} + m^{[0}m^{*c]}) + \Phi_2 \ell^{[0}m^{c]}) \big) \star K \label{eta-rep}\\
\lambda=&\,  \ptl_b \big(2\bar{\mu}_q q^{ab} \textnormal{Re} (\Phi_0 m^*_{[a}n_{3]} + \Phi_1 (n_{[a} \ell_{3]} + m_{[a} m^*_{3]})
+\Phi_2 \ell_{[a} m_{3]})  \big)\star K, \label{lambda-rep}
\end{align}
respectively, where $K$ is the fundamental solution of the $2-$Laplacian and $\star$ is the convolution in $\mbo{\Sigma}$ with the flat metric.

\item There exits a $1-$parameter family of positive-definite and conserved energy functionals for the initial value problem of the Maxwell scalars $\Phi_0, \Phi_1, \Phi_2$
\eqref{Maxwell-eqn}
\end{enumerate}

\begin{proof}
In this work we shall use the coordinate system $(\bar{\r}, \bar{z})$ on $(\Sigma, q)$ such that $\mathcal{H}^+ \cup \Gamma = \{\bar{\r} =0\}$ but the results extend to other coordinates (cf. Appendices G and H in \cite{GM17} ). Likewise, we shall restrict to $\eta$ and the proof is similar for $\lambda$.  It follows from the definition of $\eta$ and the regularity conditions on the axes $\Gamma$ that 
\begin{subequations} 
\begin{align}
\leftexp{(2)}{\Delta}\, \eta =&\, \ptl_a \left(\bar{\mu}_q q^{ab} \eps_{cb}  \mathfrak{E}^c \right), \quad (\mbo{\Sigma}, q) \\
\eta = &0, \quad \Gamma \\
\eta = &0, \quad \mathcal{H}^+
\end{align}
\end{subequations}
with $\ptl_a \mathfrak{E}^a =0, (\mbo{\Sigma}, q),$
where 
$\displaystyle \leftexp{(2)}{\Delta}= \frac{\ptl^2}{\ptl \bar{r}^2}  + \frac{1}{\bar{r}} \frac{\ptl}{\ptl \bar{r}} + \frac{1}{\bar{r}^2} \frac{\ptl^2}{\ptl \bar{\theta}}, \quad \bar{\r} = \bar{r} \cos{\bar{\theta}},\, \bar{z} = \bar{r} \sin{\bar{\theta}}
.$
It follows from the method of images that the fundamental of solution $K$ of the Laplacian on $\Sigma$ 
\begin{subequations} 
\begin{align}
\leftexp{(2)}{\Delta}\, K =&\, \delta (\bar{r}), \quad \mbo{\Sigma}\\
 K = &0,\quad  \Gamma \\
K =& 0,\quad \mathcal{H}^+
\end{align}
\end{subequations}
is 
\begin{align}
K = \frac{1}{2 \pi} \log \varrho - \frac{1}{2 \pi} \log \varrho'
\end{align}
where $\varrho, \varrho'$ are the (Euclidean) distances from $(\bar{\r}, \bar{z}) \in \mbo{\Sigma}$ and its   `image point' $(-\bar{\r}, \bar{z})$ respectively and $\delta(\bar{r})$ is the Dirac delta function on $\mbo{\Sigma}$ with flat metric. 
$K$ has faster decay rate than the fundamental solution of the Laplacian on $\mathbb{R}^2$. Likewise, we can represent $\lambda$ as follows 
\begin{subequations} 
\begin{align}
\leftexp{(2)}{\Delta}\, \lambda =&\, \ptl_a \left(\bar{\mu}_q q^{ab} \eps_{cb}  \mathfrak{B}^c \right), \quad (\mbo{\Sigma}, q) \\
\lambda = &0, \quad \Gamma \\
\lambda = &0, \quad \mathcal{H}^+.
\end{align}
\end{subequations}

The representation formulas \eqref{eta-rep} and \eqref{lambda-rep} follow immediately.  It may be noted that, in a strict sense, the representation formulas for $\eta$ and $\lambda$ correspond to their definitions only if the Gauss constraint equations are satisfied. In our work, we are only interested in the Maxwell scalars which satisfy the Maxwell equations, so this condition is automatically satisfied. Now, eliminating the variables in $X$ in favour of the Maxwell scalars $\Phi_0, \Phi_1, \Phi_2$, we get a positive-definite energy expression for their dynamics:
 \begin{align}\label{NP-Max-Energy}
&H^{\text{NPM}}(\Phi_0, \Phi_1, \Phi_1) \notag\\
\fdg= \int_{\mbo{\Sigma}} \Bigg \{ & \halb N e^{2\gamma} \bar{\mu}^{-1}_q \Big \{  \Big(2e^{-2\gamma}N \bar{\mu}_{q} \big( \text{Re} ( \Phi_0 m^{*[0} n^{3]}+ \Phi_1 (n^{[0} \ell^{3]} + m^{[0} m^{*3]})\notag\\ 
+&\Phi_2 \ell^{[0} m^{3]}  )\big) \Big) ^2 
 +  \Big(2\text{Re} (\Phi_0 m^{*}_{[1} n_{2]}+ \Phi_1 (n_{[1} \ell_{2]} + m_{[1} m_{2]}^{*}) + \Phi_2 \ell_{[1} m_{2]}  )\Big)^2 \Big \} \notag\\
+&\halb N \bar{\mu}_q q^{ab} e^{-2\gamma} \Big \{ \Big(2 \text{Re}(\Phi_0 m^*_{[a}n_{3]}+  \Phi_1 (n_{[a} \ell_{3]} + m_{[a} m^*_{3]}) + \Phi_2 \ell_{[a} m_{3]} ) \notag\\
-&2s\ptl_a \gamma  \ptl_c \big(2\bar{\mu}_q q^{dc} \textnormal{Re} (\Phi_0 m^*_{[d}n_{3]} + \Phi_1 (n_{[d} \ell_{3]} + m_{[d} m^*_{3]}) +\Phi_2 \ell_{[d} m_{3]})  \big)\star K \notag\\
-&(1-s)e^{-2\gamma}\ptl_a \omega  \ptl_c \big(2Ne^{2\gamma} \bar{\mu}^2_q q^{dc} \eps_{fd} \textnormal{Re} (\Phi_0 m^{*[0} n^{f]} + \Phi_1 (n^{[0} \ell^{f]} + m^{[0}m^{*f]}) \notag\\+& \Phi_2 \ell^{[0}m^{f]}) \big) \star K \Big)\cdot
\Big(2\text{Re}(\Phi_0 m^*_{[b}n_{3]}+ \Phi_1 (n_{[b} \ell_{3]} + m_{[b} m^*_{3]}) + \Phi_2 \ell_{[b} m_{3]} ) \notag\\
-&2s\ptl_b \gamma  \ptl_c \big(2\bar{\mu}_q q^{dc} \textnormal{Re} (\Phi_0 m^*_{[d}n_{3]} + \Phi_1 (n_{[d} \ell_{3]} + m_{[d} m^*_{3]}) +\Phi_2 \ell_{[d} m_{3]})  \big)\star K \notag\\
-&(1-s)e^{-2\gamma}\ptl_b \omega  \ptl_c \big(2Ne^{2\gamma} \bar{\mu}^2_q q^{dc} \eps_{fd} \textnormal{Re} (\Phi_0 m^{*[0} n^{f]} + \Phi_1 (n^{[0} \ell^{f]} + m^{[0}m^{*f]}) \notag\\+& \Phi_2 \ell^{[0}m^{f]}) \big) \star K \Big) 
 + \Big( 2 e^{-2\gamma}N \bar{\mu}_q \eps_{ac} \text{Re}( \Phi_0 m^{[*0} n^{c]} + \Phi_1 (n^{[0}\ell^{c]}+ m^{[0} m^{*c]}) \notag\\ 
 +& \Phi_2 \ell^{[0}m ^{c]} ) 
 -2(1-s)\ptl_a \gamma \ptl_c \big(2Ne^{2\gamma} \bar{\mu}^2_q q^{dc} \eps_{fd} \textnormal{Re} (\Phi_0 m^{*[0} n^{f]} + \Phi_1 (n^{[0} \ell^{f]} + m^{[0}m^{*f]}) \notag\\
 +& \Phi_2 \ell^{[0}m^{f]}) \big) \star K 
 +se^{-2\gamma}\ptl_a \omega \ptl_c \big(\bar{\mu}_q q^{dc} \textnormal{Re} (\Phi_0 m^*_{[d}n_{3]} + \Phi_1 (n_{[d} \ell_{3]} + m_{[d} m^*_{3]}) \notag\\
 +&\Phi_2 \ell_{[d} m_{3]})  \big)\star K \Big) 
 \cdot \Big( 2 e^{-2\gamma}N \bar{\mu}_q \eps_{bc} \text{Re}( \Phi_0 m^{[*0} n^{c]} + \Phi_1 (n^{[0}\ell^{c]}+ m^{[0} m^{*c]}) + \Phi_2 \ell^{[0}m ^{c]} )  \notag\\
  -&2(1-s)\ptl_b \gamma \ptl_c \big(2Ne^{2\gamma} \bar{\mu}^2_q q^{dc} \eps_{fd} \textnormal{Re} (\Phi_0 m^{*[0} n^{f]} + \Phi_1 (n^{[0} \ell^{f]} + m^{[0}m^{*f]}) + \Phi_2 \ell^{[0}m^{f]}) \big) \star K  \notag\\
 +&se^{-2\gamma}\ptl_b \omega \ptl_c \big(\bar{\mu}_q q^{dc} \textnormal{Re} (\Phi_0 m^*_{[d}n_{3]} + \Phi_1 (n_{[d} \ell_{3]} + m_{[d} m^*_{3]})
+\Phi_2 \ell_{[d} m_{3]})  \big)\star K \Big) \Big \} \notag\\
  +&2s(1-s) N \bar{\mu}_q q^{ab} (\ptl_a \gamma \ptl_b \gamma + \frac{1}{4}e^{-4\gamma} \ptl_a \omega \ptl_b \omega) \cdot \notag\\
 \Big(& (\ptl_b \big(2\bar{\mu}_q q^{ab} \textnormal{Re} (\Phi_0 m^*_{[a}n_{3]} + \Phi_1 (n_{[a} \ell_{3]} + m_{a} m^*_{3]})
+\Phi_2 \ell_{[a} m_{3]})  \big)\star K)^2 \notag\\ 
+&(\ptl_b \big(2Ne^{2\gamma} \bar{\mu}^2_q q^{ab} \eps_{ca} \textnormal{Re} (\Phi_0 m^{*[0} n^{c]} + \Phi_1 (n^{[0} \ell^{c]} + m^{[0}m^{*c]}) + \Phi_2 \ell^{[0}m^{c]}) \big) \star K )^2\Big) \Bigg\} d^2x.
\end{align}

\end{proof}
\end{theorem}
In contrast with \eqref{Bel-Rob-Max}, the energy functional \eqref{NP-Max-Energy} is nonlocal in the Maxwell scalars $\Phi_0, \Phi_1, 
\Phi_2.$ If desired, the $\ptl \omega$ terms can be eliminated using \eqref{grav-twist} to obtain a completely 3+1 representation of the energy functional \eqref{NP-Max-Energy}.

\subsection*{Acknowledgements} I express my gratitude to Vincent Moncrief for the enjoyable discussions and the feedback. 

\bibliography{../../References/central-bib.bib}
\bibliographystyle{plain}

\end{document}